\documentclass[a4paper,11pt]{article}
 
\usepackage{xspace}
\usepackage{comment}
\usepackage{amsmath}
\usepackage{amsfonts}
\usepackage{amssymb}
\usepackage{amsthm}
\usepackage{fullpage}
\usepackage{graphicx}

\newtheorem{lemma}{Lemma}[section]

\newtheorem{theorem}[lemma]{Theorem}

\theoremstyle{definition}

\newcommand{\dist}{\delta}
\newcommand{\cost}{\mathrm{cost}}
\newcommand{\val}{\mathrm{value}}
\newcommand{\cP}{\mathcal{P}}
\newcommand{\cC}{\mathcal{C}}
\newcommand{\cS}{\mathcal{S}}
\newcommand{\cB}{\mathcal{B}}
\newcommand{\eps}{\varepsilon}

\title{On Pairwise Spanners\footnote{Partially supported by the ERC Starting Grant NEWNET 279352. This work was partially done while the third author was visiting IDSIA.}}

\author{
 Marek Cygan \\
  IDSIA, University of Lugano, Switzerland \\
\small  \texttt{marek@idsia.ch}  \\
\and
Fabrizio Grandoni \\
  IDSIA, University of Lugano, Switzerland \\
\small   \texttt{fabrizio@idsia.ch}  \\
\and
Telikepalli Kavitha \\
Tata Institute of Fundamental Research, India \\
\small \texttt{kavitha@tcs.tifr.res.in}
}

\begin{document}

\maketitle

\begin{abstract}
Given an undirected $n$-node unweighted graph $G = (V, E)$, a spanner with \emph{stretch function} $f(\cdot)$ is a subgraph $H\subseteq G$ such that, if two nodes are at distance $d$ in $G$, then they are at distance at most $f(d)$ in $H$. Spanners are very well studied in the literature. The typical goal is to construct the sparsest possible spanner for a given stretch function. 

In this paper we study \emph{pairwise spanners}, where we require to approximate the $u$-$v$ distance 
only for pairs $(u,v)$ in a given set $\cP \subseteq V\times V$.
Such $\cP$-spanners were studied before [Coppersmith,Elkin'05] only in the special case that $f(\cdot)$ is the identity function, i.e. distances between relevant pairs must be preserved exactly (a.k.a. {\em pairwise preservers}). 

Here we present pairwise spanners which are at the same time sparser than the best known preservers (on the same $\cP$) and of the best known spanners (with the same $f(\cdot)$).
In more detail, for arbitrary $\cP$, we show that there exists a $\mathcal{P}$-spanner of size $O(n(|\cP|\log n)^{1/4})$ with $f(d)=d+4\log n$. Alternatively, for any $\eps>0$, there exists a $\cP$-spanner of size $O(n|\cP|^{1/4}\sqrt{\frac{\log n}{\eps}})$ with $f(d)=(1+\eps)d+4$. We also consider the relevant special case that there is a critical set of nodes $S\subseteq V$, and we wish to approximate either the distances within nodes in $S$ or from nodes in $S$ to any other node. We show that there exists an $(S\times S)$-spanner of size $O(n\sqrt{|S|})$ with $f(d)=d+2$, and an $(S\times V)$-spanner of size $O(n\sqrt{|S|\log n})$ with $f(d)=d+2\log n$. All the mentioned pairwise spanners can be constructed in polynomial time.
\end{abstract}

\section{Introduction}
\label{intro}

Let $G = (V,E)$ be an undirected unweighted graph. A subgraph $H$ of $G$ is a \emph{spanner} with stretch function $f(\cdot)$ if, given any two nodes $s, t \in V$ at distance $\delta_G(s,t)$ in $G$, the distance $\delta_H(s,t)$ between the same two nodes in $H$ is at most $f(\delta_G(s,t))$. An $(\alpha,\beta)$ spanner is a spanner with stretch functions $f(d)=\alpha\cdot d+\beta$. ($\alpha$ and $\beta$ are the \emph{multiplicative stretch} and \emph{additive stretch} of the spanner, respectively). If $\beta=0$ the spanner is called \emph{multiplicative}, and if $\alpha=1$ the spanner is called \emph{purely-additive}. 
Spanners are very well studied in the literature (see Section \ref{sec:related}). The typical goal is to achieve the sparsest possible spanner for a given stretch function $f(\cdot)$ \cite{BKMP05,BS03,DHZ97,E05,EP04,HZ96,PS89,RZ04,RMZ05,TZ06}.

In this paper we address the natural problem of finding (even sparser) spanners in the case that we want to approximately preserve distances only among a given subset $\cP\subseteq V\times V$ of pairs. More formally a \emph{pairwise spanner} on pairs $\cP$, or $\cP$-spanner for short, with stretch function $f(\cdot)$ is a subgraph $H\subseteq G$ such that, for any $(s,t)\in \cP$, $\delta_H(s,t)\leq f(\delta_G(s,t))$. In particular, a classical (all-pairs) spanner is a $(V\times V)$-spanner. Pairwise spanners capture scenarios where we only (or mostly) care about some distances in the graph. 

To the best of our knowledge, pairwise spanners were studied before only in the special case that $f(\cdot)$ is the identity function, i.e.  distances between relevant pairs have to be preserved exactly.  Coppersmith and Elkin \cite{CE05} call such spanners \emph{pairwise (distance) preservers}, and show that one can compute pairwise preservers of size (i.e., number of edges) $O(\min\left\{|\cP|\sqrt{n}, \ n\sqrt{|\cP|}\right\})$. 

The authors left it as an open problem to study the {\em approximate} variants of these 
preservers, i.e. what we call pairwise spanners here. This paper takes the first step in answering this question. We show that (for suitable $\cP$) it is possible to achieve $\cP$-spanners which are at the same time sparser than the preservers in \cite{CE05} (on the same set $\cP$) and than the sparsest known classical spanners (with the same stretch function).


\subsection{Our Results and Techniques}

In this paper we present some polynomial-time algorithms to construct $(\alpha,\beta)$ $\cP$-spanners for unweighted graphs. Our spanners are either purely-additive (i.e. $\alpha=1$) or \emph{near-additive} (i.e. $\alpha=1+\eps$ for an arbitrarily small $\eps>0$). 
For arbitrary $\cP$, we achieve the following main results (see Section \ref{section-pairwise}).
\begin{theorem}
\label{thm:pairwise} {\bf (near-additive pairwise)}
For any $\eps>0$ and any $\cP\subseteq V\times V$, there is a polynomial time algorithm to compute a $(1+\eps,4)$ $\cP$-spanner of size $O(n|\cP|^{1/4}\sqrt{\log n / \eps})$. 
\end{theorem}
\begin{theorem}
\label{thm:pairwise2} {\bf (purely-additive pairwise)}
For any integer $k\geq 1$ and any $\cP\subseteq V\times V$, there is a polynomial time algorithm to compute a $(1,4k)$ $\cP$-spanner of size\newline $O(n^{1+1/(2k+1)}(\sqrt{(4k+5)|\cP|})^{k/(2k+1)})$.
\end{theorem}

We also consider the relevant special case that all the pairs involve at least one node from a critical set $S\subseteq V$. More precisely, we distinguish two types of such pairwise spanners: in \emph{subsetwise spanners} (see Section \ref{section-s-spanners}) we wish to approximate distances \emph{between} nodes in $S$, i.e. $\cP=S\times S$; in \emph{sourcewise spanners} (see Section \ref{section-sourcewise}) we wish to approximate distances \emph{from} nodes in $S$, i.e. $\cP=S\times V$. We obtain the following improved results for the mentioned cases. 

\begin{theorem}
\label{thm:s-spanners} {\bf (subsetwise)}
For any $S\subseteq V$, there is a polynomial time algorithm to compute a $(1,2)$ $(S\times S)$-spanner of size $O(n\sqrt{|S|})$. 
\end{theorem}
\begin{theorem}
\label{thm:sourcewise} {\bf (sourcewise)}
For any $S \subseteq V$ and any integer $k \ge 1$, there is a polynomial time algorithm to compute a $(1,2k)$ $(S\times V)$-spanner of size $O(n^{1+1/(2k+1)}(k|S|)^{k/(2k+1)})$ .
\end{theorem}

In particular, by choosing $k= \log n$, we obtain a $(1,2\log n)$ sourcewise spanner of size  $O(n\sqrt{|S|\log n})$, and a 
$(1,4\log n)$ pairwise spanner of size
$O(n(|\cP|\log n)^{1/4})$.

All our spanners rely on a path-buying strategy which was first exploited in the $(1,6)$ spanner by Baswana et al. \cite{BKMP05}.
The high-level idea is as follows. There is an initial clustering phase, where we compute a suitable clustering of the nodes, and an associated subset of edges which are added to the spanner. Then there is a path-buying phase, where we consider an appropriate sequence of paths, and decide whether to add or not each path in the spanner under construction\footnote{In the spanner from Theorem \ref{thm:pairwise} there is also a final step where we add a  multiplicative $(2\log n,0)$-spanner.}. In particular, each path has a \emph{cost}  which is given by the number of edges of the path not already contained in the spanner, and a \emph{value} which measures \emph{how much} the path helps to satisfy the considered set of constraints on pairwise distances. If the value is sufficiently larger than the cost, we add the considered path to the spanner, otherwise we do not.

In more detail, all our pairwise spanners exploit the same clustering phase. We compute a partition $\cC=\{C_1,\ldots,C_q\}$ of a subset of the nodes, and call \emph{unclustered} the remaining nodes $V-\cup_i C_i$. The initial value of the spanner is $G_{\cC}=(V,E_{\cC})$, where $E_{\cC}$ contains all the edges of $G$ but possibly a subset of the \emph{inter-cluster} edges (with endpoints in two different clusters).
The common clustering phase is described in Section \ref{section-clustering}.


During the path-buying phase we add to the spanner some extra \emph{inter-cluster} edges. Here we need to finely tune the sequence of paths that we consider, and also the definition of value of a path. In our subsetwise and sourcewise spanners the value of a path $\rho$ reflects the number of pairs $(v,C)$, where $v$ is the endpoint of some pair and $C$ is a cluster, such that adding $\rho$ to the current spanner decreases the distance between $v$ and (the closest node in) $C$. In the remaining pairwise spanners, we use a similar notion of value, but considering the distance between pairs of clusters $(C',C'')$. 

The sequence of paths used in our subsetwise spanner and near-additive pairwise spanner is simply given by the shortest paths among the relevant pairs. This naturally generalizes the set of paths considered in \cite{BKMP05}. However, for the sourcewise spanner and the purely-additive pairwise spanner we need to consider a carefully constructed sequence of paths, which includes slightly suboptimal paths. In more detail, we start with the set of shortest paths between the relevant pairs. Then, for each such path $\rho$, if the cost of $\rho$ is sufficiently smaller than its value, we include $\rho$ in the spanner. Otherwise, we replace $\rho$ with a \emph{slightly longer} path $\rho'$ between the same endpoints which is \emph{much cheaper}, and iterate the process on $\rho'$. After a small number of iterations, the considered path becomes cheap enough and hence we include it in the spanner. 
This (non-trivial) iterative construction of candidate paths during the path-buying phase is probably the main algorithmic contribution of this paper.

\subsection{Related Work}
\label{sec:related}

Graph spanners were introduced by Peleg and Schaffer~\cite{PS89} in 1989. Spanners have been extensively 
studied since then, and there are numerous applications involving spanners, such as algorithms for 
approximate shortest paths~\cite{ABCP98,C93,E05}, labeling schemes~\cite{P00,GPPR01}, 
approximate distance oracles~\cite{TZ01,BS04,BK06}, 
routing~\cite{AP92,C01,CW04}, and network design~\cite{PU89}. 

There are several algorithms for computing multiplicative and additive spanners in weighted and 
unweighted graphs. In unweighted graphs, for any integer $k \ge 1$, Halperin and Zwick~\cite{HZ96} 
gave a linear time algorithm to compute a multiplicative $(2k-1,0)$-spanner of size $O(n^{1+1/k})$, 
where $n$ is the number of vertices. 
Note that for $k=\log n$ one obtains a spanner with multiplicative stretch $O(\log n)$ and with size $O(n)$: we will use this type of spanner in Theorem \ref{thm:pairwise}. 
Analogous results are also known for weighted graphs~\cite{BS03,RMZ05,RZ04}.

The first purely-additive spanner (for unweighted graphs) is due to Dor et al.~\cite{DHZ97}. They describe a $(1,2)$ spanner of 
size $O(n^{3/2}\log n)$. This was subsequently improved to $O(n^{3/2})$~\cite{EP04}. Note that our subsetwise spanner from Theorem \ref{thm:s-spanners} generalizes this result: in particular, it has the same stretch function and is sparser whenever $|S|=o(n)$.
Baswana et al. \cite{BKMP05} describe a $(1,6)$-spanner of size $O(n^{4/3})$. Whenever $|\cP|=O(n^{4/3-\delta})$ for some constant $\delta>0$, we achieve an asymptotically sparser pairwise spanner with constant additive stretch (depending on $\delta$). The same holds for our sourcewise spanner if $|S|=O(n^{2/3-\delta})$.


The result in \cite{HZ96} shows an elegant trade-off between the size of the spanner and its multiplicative stretch. No such trade-off is known for purely-additive spanners. In particular, the spanner in \cite{BKMP05} is the sparsest known purely-additive spanner. Theorems \ref{thm:pairwise2} and \ref{thm:sourcewise} show a non-trivial trade-off between the size and additive stretch of pairwise spanners. 

There have also been several results on near-additive spanners~\cite{EP04,E05,TZ06}. For example, there is a $(1+\epsilon, 4)$-spanner of size $O(\frac{n^{4/3}}{\epsilon})$ for any 
$\epsilon > 0$~\cite{EP04}. Our pairwise spanner from Theorem \ref{thm:pairwise} has the same stretch function, and is sparser for $|\cP|=o(n^{4/3}/(\eps\,\log n)^2)$. 

Compared to the preservers in \cite{CE05}, we achieve sparser  pairwise spanners with additive stretch $O(\log n)$ for $|\cP|=\omega(n^{2/3}\log^{1/3} n)$, and a sparser subsetwise spanners for $|S|=\omega(n^{1/3})$. 
Interestingly, our sourcewise spanners are always sparser than the pairwise preservers from \cite{CE05}. 

\section{Clustering}
\label{section-clustering}

A \emph{clustering} of a graph $G=(V,E)$ is a collection $\cC=\{C_1,\ldots,C_q\}$ of pairwise disjoint subsets of nodes $C_i\subseteq V$. Note that we do not require $\cC$ to span all the nodes $V$: we call \emph{unclustered} the nodes $V-\cup_i C_i$.

We will crucially exploit the following construction of a clustering $\cC$ and of an associated \emph{cluster subgraph} $G_\cC$.

\begin{lemma}
\label{lem:clustering}
There is a polynomial time algorithm which,
given $\beta \in [0,1]$ and a graph $G=(V,E)$, computes a clustering $\cC$ with at most $n^{1-\beta}$ clusters and a subgraph $G_\cC$ of size $O(n^{1+\beta})$ such that:
\begin{enumerate}
\item \label{prop:clustering:prop2} {\bf (missing-edge property)} If an edge $uv \in E$ is absent in $G_\cC$, then $u$ and $v$ belong to two different clusters.
  \item \label{prop:clustering:prop3} {\bf (cluster-diameter property)} The distance in $G_\cC$ between any two vertices of the same cluster is at most $2$.
\end{enumerate}
\end{lemma}
\begin{proof}
Let $U$ be the set of nodes which are not yet clustered (initially we set $U:=V)$.
As long as there exists a vertex $v \in V$ with at least $\lceil n^\beta\rceil$ neighbors in $U$, let $C$ contain exactly $\lceil n^\beta \rceil$ arbitrary
neighbors of $v$ in $U$.
Add $C$ to $\cC$, set $U := U\setminus C$ and add to $G_\cC$ 
all the edges of $G$ with both endpoints in $C \cup \{v\}$. 
When no node $v$ satisfies the mentioned property, we stop creating new clusters and add to $G_\cC$ all the edges incident to the final set of unclustered nodes $U$.

By construction, clusters are pairwise disjoint. Each time we create a new cluster, the size of $U$ decreases by at least $n^\beta$, hence there cannot be more than $n^{1-\beta}$ clusters. Any two nodes in the same cluster $C$ have some common neighbor $v$ in $G_\cC$, hence Property \ref{prop:clustering:prop3} is satisfied. By construction, all the edges incident to unclustered nodes plus the intra-cluster edges (with both endpoints in the same cluster) belong to $G_{\cC}$, which implies Property \ref{prop:clustering:prop2}.

It remains to bound the number of edges of $G_{\cC}$. Each time we create a new cluster, the number of edges of $G_{\cC}$ grows by at most $O(n^{2\beta})$: this gives $O(n^{1-\beta}n^{2\beta})=O(n^{1+\beta})$ edges altogether. When we stop creating clusters, each (clustered or unclustered) node $v$ has at most $n^{\beta}$ neighbors in $U$: consequently the number of edges incident to unclustered nodes that we add at the end of the procedure is at most $O(n^{1+\beta})$.  
\end{proof}

The following technical lemma turns out to be useful in the remaining sections.
\begin{lemma}
\label{lem:num-clusters}
Let $\cC$ and $G_\cC$ be constructed with the procedure from Lemma \ref{lem:clustering} w.r.t. a given graph $G=(V,E)$. If the shortest path $\rho$ in $G$ between any two nodes $u,v\in V$ contains $t$ edges that are absent in $G_\cC$, 
then there are at least $t/2$ clusters of $\cC$
having at least one vertex on $\rho$.
\end{lemma}

\begin{proof}
We prove the lemma by counting pairs $(u,e)$, where $e$ is an edge of $\rho$
absent in $G_\cC$ and $u$ is one of the endpoints of $e$: let $\cS$ be the set of such pairs.
Since $\rho$ contains $t$ edges that are absent in $G_\cC$ there are exactly $2t$
pairs in $\cS$ (each edge $e = uv$ belongs to two pairs: 
$(u,e)$ and $(v,e)$).
We say that a cluster $C \in \cC$ {\em owns} a pair $(u,e)$ if $u \in C$.
By the missing-edge property, each edge $e$ of $\rho$ absent in $G_\cC$ has both endpoints clustered,
hence each pair of $\cS$ is owned by some cluster.

Let us assume that there are $x$ clusters of $\cC$ having
at least one vertex on $\rho$.
By the cluster-diameter property any cluster $C \in \cC$ contains at most $3$ vertices on $\rho$,
since otherwise $\rho$ would not be a shortest path between $u$ and $v$.
However, if a cluster $C \in \cC$ contains exactly $3$ vertices on $\rho$,
those have to be consecutive vertices $a,b,c$ of $\rho$, since $\rho$ is a shortest path and we know by the cluster-diameter property that there is a path of length at most 2 between every 
pair in $\{a,b,c\}$. By the missing-edge property
both edges $ab$ and $bc$ are present in $G_\cC$, and consequently $C$
owns at most two pairs of $\cS$.
Clearly if a cluster $C \in \cC$ contains at most $2$ vertices on $\rho$,
then it owns at most $4$ pairs of $\cS$.
Therefore each cluster owns at most $4$ pairs of $\cS$:  
since $\cS$ has $2t$ pairs we have $x \ge t/2$.
\end{proof}

\section{Subsetwise Spanners}
\label{section-s-spanners}

In this section we present our algorithm to compute a subsetwise spanner, and prove Theorem~\ref{thm:s-spanners}. 

Our algorithm 
consists of two main phases: a clustering phase and a path-buying phase. In the clustering phase we invoke Lemma~\ref{lem:clustering} 
and obtain a cluster subgraph $G_\cC$ of $G$ of size $O(n^{1+\beta})$,
together with a set $\cC$ of at most $n^{1-\beta}$ clusters. The value of $\beta$ will be defined later.

In the path-buying phase we proceed as follows. 
Initially set $G_0:=G_\cC$ and let $\{\rho_1, \ldots, \rho_z \}$,  denote
the set of $z=\binom{|S|}{2}$ shortest paths between all pairs of vertices in $S$. We let $(u_i,v_i)$ denote the endpoints of $\rho_i$.
Next, we iterate over the paths $\rho_i$ for $i=1,\ldots,z$. 
To determine which paths are affordable, we define the functions $\val(\cdot)$
and $\cost(\cdot)$:
\begin{itemize}
\item let $\cost(\rho_i)$ be the number of edges of $\rho_i$ that are absent in $G_{i-1}$
\item let $\val(\rho_i)$ be the number of pairs $(x,C)$, where $x \in \{u_i, v_i\}$ and 
$C \in \cC$ is a cluster, such that $\rho_i$ contains at least one vertex of $C$
and the distance between $x$ and $C$ in the graph $G_{i-1}$
is strictly greater than the distance between $u$ and $C$ in $\rho_i$, 
i.e., $\delta_{G_{i-1}}(x,C) > \delta_{\rho_i}(x,C)$.
\end{itemize}



Our path-buying strategy is as follows. If 
$$\cost(\rho_i)\leq 2\,\val(\rho_i)$$ 
then we buy the path $\rho_i$, that is we set $G_i:=G_{i-1}\cup \rho_i$ 
(in words, $G_i$ is given by $G_{i-1}$ plus the edges of $\rho_i$ not in $G_{i-1}$).
Otherwise (i.e., $2\val(\rho_i) < \cost(\rho_i)$), we do not buy $\rho_i$
and set $G_i:=G_{i-1}$.
The subsetwise spanner is given by $G_s:=G_z$.

The next two lemmas bound the stretch and the size of the constructed spanner $G_s$, respectively

\begin{lemma}
\label{lem:s-1}
For any $(u_i,v_i) \in \cP$,  $\dist_{G_s}(u_i,v_i) \le \dist_G(u_i,v_i)+2$.
\end{lemma}

\begin{proof}
Clearly the claim holds if our algorithm bought the path $\rho_i$,
hence we assume $2\val(\rho_i) < \cost(\rho_i)$.
Let $\cost(\rho_i)=t$, that is there are exactly $t$
edges of $\rho_i$ which are not present in the graph $G_{i-1}$.
By Lemma~\ref{lem:num-clusters} there are at least $t/2$ clusters having
at least one vertex on $\rho_i$.
If there is no cluster $C$ among them such that 
$\dist_{G_{i-1}}(u_i,C) = \dist_{G}(u_i,C)$ and
$\dist_{G_{i-1}}(v_i,C) = \dist_{G}(v_i,C)$,
then all these clusters would contribute to $\val(\rho_i)$ (either with $u_i$ or with $v_i$ or both)
which leads to a contradiction, because $t = 2 \cdot (t/2) \le 2\val(\rho_i) < \cost(\rho_i) = t$.

Thus there is a cluster $C$ having a vertex of $\rho_i$ such that 
$\delta_{G_{i-1}}(u_i,C)=\delta_{G}(u_i,C)$ and
$\delta_{G_{i-1}}(v_i,C)=\delta_{G}(v_i,C)$. This implies:
\begin{eqnarray*}
\delta_{G_s}(u_i,v_i) \ \le \ \delta_{G_{i-1}}(u_i,v_i) & \le & \delta_{G_{i-1}}(u_i,C) + \delta_{G_{i-1}}(v_i,C)+2 \\
& \le & \delta_{G}(u_i,C)+\delta_{G}(v_i,C)+2\\
&\le  & \delta_G(u_i,v_i)+2,
\end{eqnarray*}
where the first inequality is because $G_{i-1}$ is a subgraph 
of $G_s$,
the second inequality holds since any two vertices of $C$ are at distance at most two in $G_\cC \subseteq G_{i-1}$ (by the cluster-diameter property)
and the last inequality follows from the assumption that $C$ contains a vertex of $\rho_i$.
\end{proof}

\begin{lemma}
\label{lem:s-2}
For $\beta$ such that $n^{\beta}=\sqrt{|S|}$ 
the graph $G_s$ contains $O(n\sqrt{|S|})$ edges.
\end{lemma}

\begin{proof}
The clustering phase produces a graph with $O(n^{1+\beta})=O(n\sqrt{|S|})$ edges.
Let $\cB$ be the set of paths bought in the path-buying phase.
The total number of edges that appear in $\cB$ and do not appear in $G_\cC$
is equal to $\sum_{\rho_i \in \cB} \cost(\rho_i)$, which is 
upper bounded by $\sum_{\rho_i \in \cB} 2\val(\rho_i)$.
Observe that after the first contribution of a pair $(x,C)$
to the above sum, the distance between $x$ and $C$ is at most $\dist_G(x,C)+2$,
hence each pair $(x,C)$ can contribute to the sum at most $3$ times.
Therefore the total number of edges added in the second phase of our algorithm
is upper bounded by $O(n^{1-\beta}|S|)=O(n\sqrt{|S|})$.
\end{proof}

The proof of Theorem~\ref{thm:s-spanners} follows from Lemmas~\ref{lem:s-1} and \ref{lem:s-2}.


\section{Sourcewise spanners}
\label{section-sourcewise}

In this section we present our algorithm to compute a sourcewise spanner from sources $S$, and prove Theorem \ref{thm:sourcewise}. 

Our algorithm again consists of two phases, where the first 
is a clustering phase and the second is a path-buying phase.
The clustering phase is as in the algorithm from previous section, for a proper value of $\beta$ to be defined later. Let $\cC$ and $G_\cC$ be the resulting clustering and cluster subgraph. 

At the start of the second phase we set $G_0:=G_\cC$
and define $\{\rho_1, \ldots, \rho_z\}$ as the 
set of shortest paths between any two vertices of $V$
such that at least one of them belongs to $S$.
Let us assume that the path $\rho_i$ is a shortest path 
between $u_i \in S$ and $v_i \in V$.
Next, we iterate over paths $\rho_i$ for $i=1,\ldots,z$.
For a given $i$ we are going to define paths $\rho_i^j$, where $0 \le j \le k$,
maintaining the following invariants:\smallskip

\begin{itemize}
  \item[(i)] $\rho_i^j$ is a path between $u_i$ and $v_i$ in $G$ of length at most $\dist_G(u_i,v_i)+2j$,
  \item[(ii)] any cluster $C \in \cC$ contains at most three vertices of $\rho_i^j$,
  \item[(iii)] $\cost(\rho_i^j) \le 2n^{1-\beta} / \gamma^j$, where $\cost(\rho_i^j)$ is the number of edges of $\rho_i^j$ absent in $G_{i-1}$,
and $\gamma=(3n^{1-\beta})^{1/k}$.
\end{itemize}

\smallskip\noindent Our algorithm will buy exactly one path $\rho_i^j$ for $0 \le j \le k$, which will ensure 
(by Invariant (i)) that in $G_{i}$, the distance between $u_i$ and $v_i$ is at most $\dist_G(u_i,v_i)+2k$.

We set $\rho_i^0 := \rho_i$. 
Observe that for $j=0$, Invariant (i) is trivially satisfied,
Invariant (ii) is satisfied by the cluster-diameter property  
(otherwise $\rho_i$ would not be a shortest path),
and Invariant (iii) is satisfied because there are at most $n^{1-\beta}$ clusters in $\cC$
and consequently by Lemma~\ref{lem:num-clusters} $\cost(\rho_i) \le 2n^{1-\beta}$. 

Say we have constructed $\rho_i^j$, where $j \in \{0,\ldots,k\}$.
Let us define the function $\val(\rho_i^j)$ as the number of clusters $C \in \cC$ 
such that $C$ contains a vertex of $\rho_i^j$ and the distance
between $u_i$ and $C$ in $G_{i-1}$ is strictly greater than
the distance between $u_i$ and $C$ in $\rho_i^j$, i.e., $\dist_{G_{i-1}} (u_i,C) > \dist_{\rho_i^j}(u_i,C)$.
Now we check the condition 
$$\cost(\rho_i^j) \le 3\gamma\val(\rho_i^j).$$ 

If that is the case, then we buy the path $\rho_i^j$. That is, $G_i$ is set to $G_{i-1} \cup \rho_i^j$.
We ignore the remaining values of $j$ and proceed with the next value of $i$.
Else we construct $\rho_i^{j+1}$ as follows:

Let $R$ be the longest suffix of $\rho_i^j$ containing exactly $\lfloor \cost(\rho_i^j)/\gamma \rfloor$ edges that are absent in $G_{i-1}$.
Observe that the first node of $R$ is clustered:
by the maximality of $R$, the edge $e$ of $\rho_i^j$ preceding $R$
is absent in $G_{i-1}$, and hence both the endpoints of $e$ (one of which is the first node of $R$) are clustered by the missing-edge property of $G_\cC$. 
Consequently at least $1+\lfloor \cost(\rho_i^j)/\gamma \rfloor \ge \cost(\rho_i^j)/\gamma$ 
vertices of $R$ are clustered, as $R$ contains $\lfloor \cost(\rho_i^j)/\gamma \rfloor$ edges
absent in $G_{i-1}$ and the endpoints of these edges are clustered. 

By Invariant (ii) there are at least $\cost(\rho_i^j)/(3\gamma)$ clusters in $\cC$
having at least one vertex of $R$. Since we did not buy $\rho_i^j$, there
exists a cluster $C \in \cC$ containing a vertex $x \in C$ of $R$ such that the distance between $u_i$ 
and $C$ in $G_{i-1}$ is at most the distance between $u_i$ and $x$ in $\rho_i^j$. We construct the path $\rho_i^{j+1}$ by taking
a shortest path in $G_{i-1}$ from $u_i$ to the closest node $y\in C$, then we add a path of length at most two between $y$ and $x$ (which exists in $G_\cC$ hence in $G_{i-1}$ by the cluster-diameter property), and finally add the suffix of $R$ starting at $x$ (see Fig.~\ref{fig1}).


Let us show that $\rho^{j+1}_i$ maintains the invariants. Note that by construction, Invariant (i) is satisfied, since the length of $\rho_i^{j+1}$ is at most the length of $\rho_i^j$ plus $2$.
Then, as long as there is a cluster $C \in \cC$ containing at least four vertices on $\rho_i^{j+1}$,
we let $a$, $b$ be the vertices of $\rho_i^{j+1}$ closest to $u_i$ and $v_i$ respectively.
Note that there are at least three edges on $\rho_i^{j+1}$ between $a$ and $b$, hence we
can replace the subpath of $\rho_i^{j+1}$ by adding the at most two edges of $G_\cC$
guaranteed by the cluster-diameter property.
Consequently, Invariant (ii) is satisfied. Moreover, by the choice of $R$, Invariant (iii) is also satisfied.
This finishes the construction of $\rho_i^{j+1}$. 

\begin{figure}[t]
\begin{center}
\includegraphics{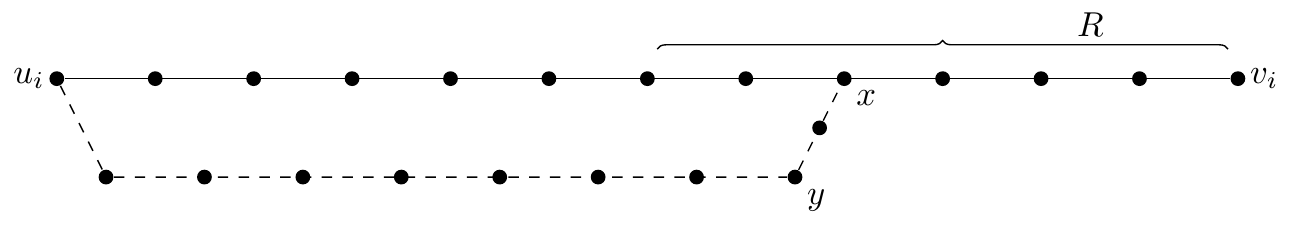}
\end{center}
\caption{The solid edges represent a path $\rho_i^j$, while the dashed edges
  denote the new prefix of the path $\rho_i^{j+1}$.}
\label{fig1}
\end{figure}

Observe that by Invariant (iii) we have $\cost(\rho_i^k) \le 2/3$: since 
$\cost(\cdot)$ has only integral values, it has to be that $\cost(\rho_i^k) = 0$,
which ensures that we buy a path $\rho_i^j$ for some $j \le k$.

Finally, as our spanner $G_s$ we take $G_s := G_z$.

%
%
%
%

\begin{lemma}\label{lem:stretchSourcewise}
For any pair $(u_i,v_i) \in \cP$,  $\dist_{G_s}(u_i,v_i) \le \dist_G(u_i,v_i)+2k$.
\end{lemma}
\begin{proof}
From the above discussion, we buy at least one path $\rho^j_i$ for some $0\leq j\leq k$. By Invariant (i), the length of the latter path is  at most the length of the shortest path $\rho_i$ between $u_i$ and $v_i$ plus $2k$.
\end{proof}

\begin{lemma}
\label{lem:sizeSourcewise}
For $\beta$ such that $n^\beta=(n^{1/k}(2k+3)|S|)^{k/(2k+1)}$, the subgraph $G_s$ contains \linebreak $O(n^{1+1/(2k+1)}(k|S|)^{k/(2k+1)})$ edges. 
\end{lemma}
\begin{proof}
To bound the size of $G_s$
we recall that in the first phase we have inserted $O(n^{1+\beta})$ edges.
Let $0 \le j_i \le k$ be the index of a path $\rho_i^{j_i}$ bought for a given $i$.
We claim, that any cluster $C$ contributes to $\val(\rho_i^{j_i})$ of at most
$|S|(2k+3)$ bought paths.
This holds because when for $u_i \in S$ a supported path is bought 
the distance between $u_i$ and $C$ is at most $2k+2$ greater than 
the distance between $u_i$ and $C$ in $G$: otherwise one could shorten $\rho_i^{j_i}$
by more than $2k$, obtaining a contradiction with Invariant (i).
Therefore the total number of edges added during the second phase is upper bounded by
$\sum_{i=1}^z \cost(\rho_i^{j_i}) \le \sum_{i=1}^z 3\gamma\val(\rho_i^{j_i}) \le 3\gamma(2k+3)|S|n^{1-\beta}$,
since each cluster $C \in \cC$ supports at most $|S|(2k+3)$ bought paths. The claim follows.
\end{proof}
The proof of Theorem \ref{thm:sourcewise} follows from Lemmas \ref{lem:stretchSourcewise} and \ref{lem:sizeSourcewise}.


\section{Pairwise spanners}
\label{section-pairwise}

In this section we present our pairwise spanners for arbitrary $\cP$. We start with a near-additive spanner (see Section \ref{section-pairwise-near}) and then present a purely-additive spanner (see Section \ref{section-pairwise-pure}). In both cases we let $\cP=\{(s_1,t_1),\ldots,(s_N,t_N)\}$ denote the set of pairs, $N=|\cP|$.


\subsection{A Near-Additive Pairwise Spanner}
\label{section-pairwise-near}

Our algorithm to construct the near-additive $\mathcal{P}$-spanner from Theorem \ref{thm:pairwise} consists of three phases. First, 
we use Lemma~\ref{lem:clustering} with the value of $\beta$
to be determined later, obtaining a cluster subgraph $G_\cC$ of $G$
of size $O(n^{1+\beta})$ together with a set $\cC$ of at most $n^{1-\beta}$ clusters.

At the start of the second phase we set $G_0 := G_\cC$ and consider the set of paths  $\{\rho_1,\ldots,\rho_N\}$, where $\rho_i$ 
is a shortest path between $s_i$ and $t_i$ in $G$.
Next we iterate over the paths $\rho_i$ for  $i=1,\ldots,N$.
By $\cost(\rho_i)$ we denote the number of edges of $\rho_i$
absent in $G_{i-1}$, and by $\val(\rho_i)$ we denote the number of 
pairs of clusters $(C_1,C_2) \in \cC$, such that both $C_1$ and $C_2$ contain at
least one vertex of $\rho_i$ and $\dist_{\rho_i}(C_1,C_2) < \dist_{G_{i-1}}(C_1,C_2)$.
For a given $i$ if
$$\cost(\rho_i) \le \frac{12\log n}{\epsilon} \sqrt{\val(\rho_i)},$$
then we buy $\rho_i$, that is we set $G_i:= G_{i-1} \cup \rho_i$. Otherwise we set $G_i:= G_{i-1}$.

In the third phase we add to $G_N$ the multiplicative $(2\log n,0)$ spanner of size $O(n)$ given in \cite{HZ96}: this way we obtain the desired spanner $G_s$.


%
%

In the following two lemmas we bound the stretch and size of $G_s$, respectively.

\begin{lemma}
\label{lem:stretchPairwise1}
For each $(s_i,t_i)\in \cP$, $\dist_{G_s}(s_i,t_i)\leq (1+\eps)\dist_G(s_i,t_i)+4$.
\end{lemma}

\begin{proof}
Clearly we can assume that the path $\rho_i$ was not bought in the second phase,
since otherwise the claim trivially holds.
Therefore 
$\cost(\rho_i) > \frac{12\log n}{\epsilon} \sqrt{\val(\rho_i)}.$

Let $\rho_i=(v_0=s_i,v_1,\ldots,v_{\ell-1},v_{\ell}=t_i)$
and let $I \subseteq \{0,\ldots, \ell\}$ be the set of all indices
$j$ such that $v_j$ is clustered.
Observe that if $|I| \le 1$, then by the missing-edge property
the whole path $\rho_i$ is present in $G_\cC$, and hence the claim holds.
Therefore denote $I = \{i_0,\ldots,i_w\}$, where $i_0 < i_1 < \ldots < i_w$ and $w \ge 1$.
Let $0 \le a \le b \le w$ be two indices, such that $v_{i_a} \in C_1$, $v_{i_b} \in C_2$ (for some $C_1, C_2 \in \cC$),
$\dist_{\rho_i}(C_1, C_2) \ge \dist_{G_{i-1}}(C_1,C_2)$ and the value of $b-a$ is maximized.
Note that such a pair of indices $a,b$ always exists, since we can take $a=b$.

Let $x=a+(w-b)$. Observe that any cluster $C \in \cC$ contains at most $3$ vertices of $V_I=\{v_{i_j} : 0 \le j \le w\}$,
since otherwise by the cluster-diameter property $\rho_i$ 
would not be a shortest $s_i$-$t_i$ path.
Therefore there are at least $x/6$ clusters $\cC'$ having at least one vertex in the set $\{v_{i_j} : 0 \le j < \lceil x/2 \rceil\}$,
and at least $x/6$ clusters $\cC''$ having at least one vertex in the set $\{v_{i_j} : w-\lceil x/2 \rceil < j \le w\}$.
However, each of the at least $(x/6)^2$ pairs of clusters in $\cC'\times \cC''$ contributes to $\val(\rho_i)$ 
since the difference between indices in the corresponding set is at least 
$w-\lceil x/2 \rceil + 1 - (\lceil x/2 \rceil - 1) > w-x$.
Therefore $w \ge \cost(\rho_i) > \frac{12 \log n}{\epsilon} \frac{x}{6}$ and hence
$x \le \frac{w\epsilon}{2\log n}$.

The latter bound on $x$ is sufficient to prove the claim. In fact, consider the path between $s_i$ and $t_i$ in $G_s$ obtained by concatenating the following paths
(as illustrated in Fig.~\ref{fig2}):
\begin{itemize}
  \item A shortest path in $G_s$ from $s_i$ to $v_{i_a}$. Note that in the prefix of $\rho_i$ between $s_i$ and $v_{i_a}$ there
  are $a+1$ clustered nodes and hence at most $a$ edges absent in $G_{i-1}$ (by the missing-edge property). Since $G_s$ contains the $(2\log n,0)$-spanner added in the third phase, each missing edge can be replaced by at path of length $2\log n$. Consequently, there is a path from $s_i$ to $v_{i_a}$ of length at most $\dist_{\rho_i}(s_i,v_{i_a})+2a\log n$ in $G_s$.
  \item A shortest path in $G_s$ from $v_{i_a}$ to $v_{i_b}$. Let $C_1, C_2 \in \cC$ be the clusters containing $v_{i_a}$ and $v_{i_b}$ respectively.
  We know that in $G_{i-1}$ there is a path from $C_1$ to $C_2$ of length at most $\dist_{\rho_i}(v_{i_a},v_{i_b})$,
  which can be extended to a path between $v_{i_a}$ and $v_{i_b}$ in $G_{i-1}$ by adding at most $4$ edges
  (by the cluster-diameter property).
  \item A shortest path in $G_s$ from $v_{i_b}$ to $t_i$. Observe that in the suffix of $\rho_i$ between $v_{i_b}$ and $t_i$
  there are at most $w-b$ edges absent in $G_{i-1}$ by the same argument as above. Hence, thanks to the $(2\log n,0)$-spanner added in the third phase, there is a path from $v_{i_b}$ to $t_i$ of length at most $\dist_{\rho_i}(v_{i_b},t_i)+(w-b)2\log n$ in $G_s$.
\end{itemize}
The resulting path is of length at most 
\begin{align*}
& \dist_{\rho_i}(s_i,v_{i_a})+2a\log n+\dist_{\rho_i}(v_{i_a},v_{i_b})+4+\dist_{\rho_i}(v_{i_b},t_i)+(w-b)2\log n \\
 = & \;
\dist_G(s_i,t_i)+2x\log n+4 \le (1+\epsilon)\dist_G(s_i,t_i)+4,
\end{align*} 
where the last inequality
follows from $x \le \frac{\epsilon w}{2 \log n}$ together with $w \le \dist_G(s_i,t_i)$.
\begin{figure}[t]
\begin{center}
\includegraphics{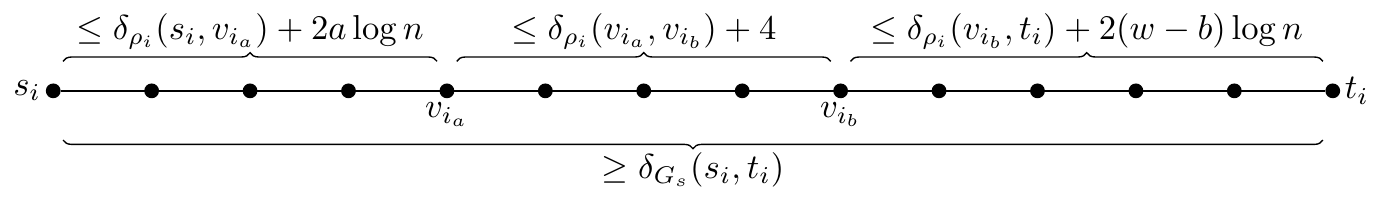}
\end{center}
\caption{Illustration of the three paths concatenation in the proof of Lemma~\ref{lem:stretchPairwise1}.}
\label{fig2}
\end{figure}
\end{proof}

\begin{lemma}
\label{lem:sizePairwise1}
For $\beta$ such that $n^{2\beta} = \sqrt{N} \frac{\log n}{\epsilon}$ the 
size of $G_s$ is $O(nN^{1/4}\sqrt{\log n / \epsilon})$.
\end{lemma}

\begin{proof}
The clustering phase gives $O(n^{1+\beta})$ edges, which matches the desired bound on $G_s$.
Let $\cB$ be the set of paths $\rho_i$ bought in the path-buying phase.
Observe, that if a pair of clusters $C_1,C_2$ contributes to $\val(\rho_i)$ of a bought path $\rho_i \in \cB$,
then when $\rho_i$ is bought we have $\dist_{G_i}(C_1,C_2) \le \dist_G(C_1,C_2)+4$, since otherwise
the subpath of $\rho_i$ between $C_1$ and $C_2$ might be shortened (by the cluster-diameter property). It follows that each pair of clusters contributes at most $5$ times to $\val(\rho_i)$, and hence 
\begin{equation}
\hspace{4cm}\sum_{\rho_i \in \cB}\val(\rho_i)\leq 5(n^{1-\beta})^2.\label{eqn:boundPairs}
\end{equation}
The total number of edges added in the second phase is therefore upper bounded by 
\begin{eqnarray*}
\sum_{\rho_i \in \cB} \cost(\rho_i) & \le & \sum_{\rho_i \in \cB} \frac{12\log n}{\epsilon} \sqrt{\val(\rho_i)} \\
   & \overset{\substack{\text{Cauchy-Schwarz}\\\text{inequality}}}{\le} & \frac{12\log n}{\epsilon} \sqrt{\sum_{\rho_i \in \cB}\val(\rho_i)} \sqrt{N} \\
  & \overset{\eqref{eqn:boundPairs}}{\le} & \frac{12\log n}{\epsilon} \sqrt{5} n^{1-\beta} \sqrt{N} \\
  & \le & 12\sqrt{5}N^{1/4}n\sqrt{\log n / \epsilon}.
\end{eqnarray*}

Finally, in the last phase we insert only $O(n)$ edges when adding the $(2\log n,0)$-spanner.
\end{proof}

Having Lemmas~\ref{lem:stretchPairwise1} and \ref{lem:sizePairwise1}, the proof of Theorem~\ref{thm:pairwise} follows.


\subsection{A Purely-Additive Pairwise Spanner}
\label{section-pairwise-pure}

In this section we describe an algorithm to compute the purely-additive $\cP$-spanner from Theorem \ref{thm:pairwise2}. To that aim we will combine ideas from the proofs of Theorems \ref{thm:pairwise} and \ref{thm:sourcewise}.

Our algorithm consists of the usual clustering phase (for an appropriate parameter $\beta$) followed by a path-buying phase that we next describe. 

Let $\cC$ and $G_\cC$ be the clustering and the associated cluster graph. At the beginning of the path-buying phase, we set $G_0 := G_\cC$ and consider the set $\{\rho_1,\ldots,\rho_N\}$,
where $\rho_i$ is a shortest path between $s_i$ and $t_i$ in $G$.
Next we iterate over the paths $\rho_i$ for $i=1,\ldots,N$.
For a given $i$ we are going to define paths $\rho_i^j$, where $0 \le j \le k$,
maintaining the following invariants:\smallskip

\begin{itemize}
  \item[(i)] $\rho_i^j$ is a path between $s_i$ and $t_i$ in $G$ of 
  length at most $\dist_G(s_i,t_i)+4j$,
  \item[(ii)] any cluster $C \in \cC$ contains at most three vertices of $\rho_i^j$,
  \item[(iii)] $\cost(\rho_i^j) \le 2n^{1-\beta} / \gamma^j$, where $\cost(\rho_i^j)$ is the number of edges of $\rho_i^j$ absent in $G_{i-1}$,
and $\gamma=(3n^{1-\beta})^{1/k}$.
\end{itemize}

\smallskip\noindent Our algorithm will buy exactly one path $\rho_i^j$, which will ensure by Invariant (i) that 
in $G_{i}$ the distance between $s_i$ and $t_i$ is at most $\dist_G(u_i,v_i)+4k$.
By $\val(\rho_i^j)$ let us denote the number 
of pairs of clusters $C_1,C_2 \in \cC$, such that both $C_1$ and $C_2$ contain at
least one vertex of $\rho_i^j$ and $\dist_{\rho_i^j}(C_1,C_2) < \dist_{G_{i-1}}(C_1,C_2)$.

We set $\rho_i^0 := \rho_i$.
Observe that for $j=0$ Invariant (i) is trivially satisfied,
Invariant (ii) is satisfied by the cluster-diameter property 
(otherwise $\rho_i$ would not be a shortest path),
and Invariant (iii) is satisfied because there are at most $n^{1-\beta}$ clusters in $\cC$
and consequently by Lemma~\ref{lem:num-clusters}, $\cost(\rho_i) \le 2n^{1-\beta}$. 

Say we have constructed $\rho_i^j$, where $j\in \{0,\ldots,k\}$.
If
$$\cost(\rho_i^j) \le 6\gamma\sqrt{\val(\rho_i^j)}\,,$$
then we buy the path $\rho_i^j$,
i.e. as $G_i$ we take the union of $G_{i-1}$ and $\rho_i^j$,
ignore remaining values of $j$ and proceed with the next value of $i$.
Otherwise (i.e., $\cost(\rho_i^j) > 6\gamma\sqrt{\val(\rho_i^j)}$), 
we construct a path $\rho_i^{j+1}$ as follows:

Let $\rho_i^j=(v_0=s_i,v_1,\ldots,v_{\ell-1},v_{\ell}=t_i)$
and let $I \subseteq \{0,\ldots, \ell\}$ be the set of all indices
$j$ such that $v_j$ is clustered.
Observe that if $|I| \le 1$, then by the missing-edge property
the whole path $\rho_i^j$ is present in $G_\cC$, and hence it 
is of zero cost, which contradicts the assumption
$\cost(\rho_i^j) > 6\gamma\sqrt{\val(\rho_i^j)}$.
Therefore denote $I = \{i_0,\ldots,i_w\}$, where $i_0 < i_1 < \ldots < i_w$ and $w \ge 1$. Let $0 \le a \le b \le w$ be two indices,
such that $v_{i_a} \in C_1$, $v_{i_b} \in C_2$ (for some $C_1, C_2 \in \cC$),
$\dist_{\rho_i^j}(C_1, C_2) \ge \dist_{G_{i-1}}(C_1,C_2)$ and the value of $b-a$ is maximized.
Note that such a pair of indices $a,b$ always exists, since we can take $a=b$.

Let $x=a+(w-b)$. By Invariant (ii) there are at least $x/6$ clusters $\cC'$ having at least 
one vertex in the set $\{v_{i_j} : 0 \le j < \lceil x/2 \rceil\}$,
and at least $x/6$ clusters $\cC''$ having at least one vertex 
in the set $\{v_{i_j} : w-\lceil x/2 \rceil < j \le w\}$.
However, each of the at least $(x/6)^2$ pairs of clusters in $\cC'\times \cC''$ contributes to $\val(\rho^j_i)$ since the difference between indices in the corresponding set is at least 
$w-\lceil x/2 \rceil + 1 - (\lceil x/2 \rceil - 1) > w-x$.
Therefore 
\begin{equation}
\cost(\rho_i^j) > 6\gamma \sqrt{\val(\rho^j_i)}\geq 6 \gamma \sqrt{\left(\frac{x}{6}\right)^2}=\gamma\,x\quad \Rightarrow \quad x  \le \cost(\rho_i^j)/\gamma.\label{eqn:pairwise2}
\end{equation}

We construct the path $\rho_i^{j+1}$ by appending the following three
paths $A$, $B$, and $C$:
\begin{itemize}
  \item As $A$ we take the prefix of $\rho_i^j$ from $s_i$ to $v_{i_a}$.
  Note that this prefix contains $a+1$ clustered nodes and hence at most $a$ edges absent 
  in $G_{i-1}$ (by the missing-edge property of $G_{\cC}$).
  \item Let $C_1, C_2 \in \cC$ be the clusters containing $v_{i_a}$ and $v_{i_b}$ respectively.
  We know that in $G_{i-1}$ there is a path from $C_1$ to $C_2$ of length at most $\dist_{\rho_i^j}(v_{i_a},v_{i_b})$,
  which can be extended to a path $B$ between $v_{i_a}$ and $v_{i_b}$ in $G_{i-1}$ by adding at most $4$ edges
  (by the cluster-diameter property).
  \item As $C$ we take the suffix of $\rho_i^j$ from $v_{i_b}$ to $t_i$,
  which contains at most $w-b$ edges absent in $G_{i-1}$ by the same argument as above.
\end{itemize}
Observe that $\rho_i^{j+1}$ contains at most $a+(w-b)=x$ edges absent in $G_{i-1}$, hence by \eqref{eqn:pairwise2} we ensure Invariant (iii).
Moreover the length of $\rho_i^{j+1}$ is at most the length of $\rho_i^j$ plus $4$, which ensures Invariant (i).
In order to ensure Invariant (ii), as long as there exists a cluster $C \in \cC$
containing at least $4$ vertices of $\rho_i^{j+1}$ we let
$u$ and $v$ be two such vertices closest to $s_i$ and $t_i$ on $\rho_i^{j+1}$ respectively
and replace the subpath of $\rho_i^{j+1}$ between $u$ and $v$ (which is of length at least three)
by a path of length at most two in $G_{i-1}$ (which exists by  the cluster-diameter property).

Observe that by Invariant (iii) we have $\cost(\rho_i^k) \le 2/3$, hence $\cost(\rho_i^k) = 0$ which ensures that we buy a path $\rho_i^j$ for some $j \le k$. Finally, as our spanner $G_s$ we take $G_s := G_N$.


\begin{lemma}\label{lem:stretchPairwise2}
For each $(s_i,t_i)\in \cP$, $\dist_{G_s}(s_i,t_i)\leq \dist_G(s_i,t_i)+4k$.
\end{lemma}
\begin{proof}
From the above discussion, $G_s$ contains at least one path $\rho^j_i$ between $s_i$ and $t_i$ for some $0\leq j\leq k$. The claim follows by Invariant (i). 
\end{proof}


\begin{lemma}\label{lem:sizePairwise2}
For $\beta$ such that $n^{\beta}=(6n^{1/k} \sqrt{(4k+5)N})^{k/(2k+1)}$ the
size of $G_s$ is\newline $O(n^{1+1/(2k+1)}(\sqrt{(4k+5)N})^{k/(2k+1)})$.
\end{lemma}

\begin{proof}
The clustering phase gives $O(n^{1+\beta})$ edges, 
which matches the desired bound on $G_s$.
Let $0 \le j_i \le k$ be the index of a path $\rho_i^{j_i}$ bought for a given $i$.
We claim, that any pair of clusters contributes 
to $\val(\rho_i^{j_i})$ of at most $(4k+5)$ bought paths.
Observe, that if a pair of clusters $C_1,C_2$ contributes to $\val(\rho_i^{j_i})$,
then when $\rho_i^{j_i}$ is bought we have
$\dist_{G_i}(C_1,C_2) \le \dist_G(C_1,C_2)+4k+4$, since otherwise
the subpath of $\rho_i^{j_i}$ between $C_1$ and $C_2$ might be shortened 
by more than $4k$, contradicting Invariant (i).
The total number of edges added in the second phase is upper bounded by 
\begin{eqnarray*}
\sum_{1 \le i \le N} \cost(\rho_i^{j_i}) & \le & \sum_{1 \le i \le N} 6\gamma\sqrt{\val(\rho_i^{j_i})} \\
   & \overset{\substack{\text{Cauchy-Schwarz}\\\text{inequality}}}{\le} & 6\gamma \sqrt{\sum_{1 \le i \le N}\val(\rho_i^{j_i})} \sqrt{N} \\
  & \le & 6\gamma \sqrt{4k+5} n^{1-\beta} \sqrt{N}\,.
\end{eqnarray*}
By substituting $\gamma$ and $\beta$ the claim follows.
\end{proof}

Theorem~\ref{thm:pairwise2} follows from Lemmas \ref{lem:stretchPairwise2} and \ref{lem:sizePairwise2}

\section{Conclusions}

We considered a natural extension to the problem of computing a sparse spanner in an undirected unweighted graph. Along with the input
graph $G = (V,E)$,  a subset $\cP \subseteq V\times V$ of relevant pairs
of vertices is also given here and we seek a sparse subgraph $H$ of $G$ 
such that for every 
pair $(u,v)$ in $\cP$, the $u$-$v$ distance $\dist_H(u,v)$ in the subgraph is close to the 
$u$-$v$ distance $\dist_{G}(u,v)$ in $G$. We showed sparse subgraphs $H$ where $\delta_H(u,v)$ is a small additive or near-additive stretch away from $\delta_G(u,v)$. 

The pairwise preservers in \cite{CE05} are at the same time more accurate and sparser than our spanners for small enough values of $|\cP|$. In particular, in that range of values of $|\cP|$ the authors exploit a construction which does not seem to benefit from allowing a larger stretch. The authors also present lower bounds on the size of any preserver, however it is unclear whether those lower bounds extend to the case of pairwise spanners (where distances have to be approximated rather than preserved). Obtaining sparser pairwise spanners for very small $|\cP|$, if possible, is an interesting open problem.


%


\end{document}